\documentclass[reqno,12pt]{amsart}
\usepackage{amssymb,amscd,amsthm,amsfonts, dsfont, amsmath,epsfig,epic,bbm,url}
\usepackage{framed}
  \graphicspath{{figures/}}
\usepackage{fullpage}
\usepackage{color}
\theoremstyle{plain}

\usepackage{mathtools}
\usepackage{amsmath}
\usepackage{amssymb,latexsym}
\usepackage{fullpage}
\usepackage{color}
\theoremstyle{plain}
\newtheorem{theorem}{Theorem}[section]

\newtheorem{lemma}{Lemma}[section]
\newtheorem{proposition}{Proposition}[section]
\theoremstyle{remark}

\newtheorem{remark}{Remark}[section]

\numberwithin{equation}{section}

\newcommand{\bC}{{\mathbb C}}

\newcommand{\bZ}{{\mathbb Z}}
\newcommand{\bN}{{\mathbb N}}

\newcommand{\End}{\operatorname{End}}

\newcommand{\B}{\mathcal{B} }

\newcommand{\fgl}{\mathfrak {gl} }

\renewcommand{\>}{\rangle}

\begin{document}

\title{Action of Clifford algebra on  the space of  sequences of  transfer operators}



\author{Natasha Rozhkovskaya}
\address{Department of Mathematics, Kansas State University, Manhattan, KS 66506, USA}
\email{rozhkovs@math.ksu.edu}

\maketitle

\begin{abstract}
We deduce from a  determinant identity 
on quantum  transfer matrices  of  generalized  
quantum  integrable  spin  chain model  their 
  generating functions. We construct the isomorphism of Clifford algebra modules of sequences  of transfer matrices and
the boson space of symmetric functions. As an application, tau-functions  of transfer matrices immediately arise from classical tau-functions 
of symmetric functions.  
\end{abstract}

\section{Introduction}
Connections between transfer matrices  of quantum integrable models  and solutions of classical  integrable  hierarchies of non-linear  partial differential  equations, observed in \cite{K1},   were developed later in  a series of papers  \cite{Z1}, \cite{Z3}, \cite{Z2}, \cite{Z4}, etc.
The authors of \cite {Z1} defined   a generating function of commuting  quantum  transfer matrices  of a generalized  
quantum  integrable  spin  chain. They  proved  that this master $T$-operator  obeys   Hirota  bilinear  equations  (see e.g. \cite{D1}, \cite{Hirota1}, \cite{Miwa1}), identifying it  with a  tau-function of the KP  hierarchy.  In \cite{Z3}, \cite{Z2}, \cite{Z4} similar identification is obtained for 
the quantum inhomogeneous XXX spin chain of $GL(N)$  type with twisted boundary conditions, quantum spin chain with trigonometric $R$-matrix, and the quantum Gaudin model with twisted boundary conditions. 

   In this  note  we consider another   generating  function  for   quantum transfer matrices of a generalized  
quantum  integrable  spin  chain.  We describe  combinatorial properties of  this generating function   and define the  action  of  Clifford algebra on a space spanned  by  sequences of  quantum transfer matrices.  This approach suggests another   interpretation of  the  connections of transfer matrices to $\tau$-function formalism, the one  that refers to the  construction of the classical  Hirota bilinear  equations  from the  vertex operator  action of Clifford algebra on the boson space of symmetric functions.  One can mention that  the master $T$-operator  of  \cite{Z1}   generalizes  the right-hand side of   the  Cauchy-Littlewood identity 
\begin{align*}
\prod_{ij}(1-a_ix_j)^{-1}=\sum_{\lambda} s_{\lambda} ({ a})s_{\lambda} ({ x}),
\end{align*}
where $s_\lambda({a})$, $s_\lambda({ x})$ are  symmetric   Schur  functions in two independent sets  of variables $a=(a_1, a_2,\dots )$
and $x=(x_1, x_2,\dots )$ (see e.g.  \cite{Md},  I.4 (4.3)),
while  the generating  function in this note  is  the   analogue of  the   generating  function for Schur functions of the form
  \begin{align}
\prod_{i<j}  \left(1-\frac{x_j}{x_i}\right) \prod_{i=1}^{l} H(x_i)=\sum_{l(\lambda)\le l} s_\lambda ({ a})  x_1^{\lambda_1}\cdots x_l^{\lambda_l},\label{fsym}
  \end{align}
where  $H(x)= \sum_k h_k({ a}) x^k$ is the generating function for  complete  symmetric  functions  $h_k({ a})$ (see e.g.  \cite{Md},  I.5, Example 29 (7)).

In Section 2 we introduce the necessary notations and definitions,  and describe  the properties of generating functions of  transfer matrix. In Section 3 we construct the action of Clifford  algebra on the space of sequences of $T$-operators. In  Section 4  we make some remarks on bosonisation of the action of Clifford algebra.

\section{ Properties of Transfer operators}
\subsection{Notations and definitions}
Let $\{e_{ij}\}_{i,j=1,\dots, N}$  be the set of  standard generators of the universal  enveloping algebra $U(\fgl_N(\bC))$.
The action of these generators   on  $\bC^N$  is given by the elementary matrices  $\{ E_{ij}\}_{i,j=1,\dots, N}$.
Let  partition $\lambda=(\lambda_1\ge \dots\ge  \lambda_l>0)$ with  $\lambda_i\in\bZ_{\ge 0}$  and the length $l\le N$  be the  highest  weight of an irreducible  finite-dimensional  $U(\fgl_N(\bC))$-representation $\pi_\lambda$  acting on the space $ V_\lambda$.  We will  use the same  notation for the corresponding  representation of  $GL_N(\bC)$.
Consider    $R$-matrix given by   a  linear  function in variable $u$   with coefficients in  $\End\,V_\lambda \otimes\End \,\bC^N  $:
\begin{align}\label{Rm}
R(u)= 1+\frac{1}{u}\sum_{ij} \pi_\lambda(e_{ji})\otimes E_{ij} . 
\end{align}
 
Let $R_{0i} (u)$  be the  operator that  acts as   $R$-matrix (\ref{Rm}) on the  $0$-th component  $V_\lambda$  and  on the $i$-th component $\bC^N$
of the  tensor  product $V_\lambda \otimes (\bC^N)^{\otimes n}$. We fix a  collection of complex parameters $ a= (a_i )_{i=1,2,\dots}$.
Let  $g$ be an invertible  $N$ by $N$ matrix,  called the  twist matrix, with eigenvalues $(g_1,\dots, g_N)$.  Consider a family of   quantum  transfer matrices ($T$-operators) 
$T_\lambda(u)=T_{\lambda, n}^N (u, a)(g)$  defined  by
 \begin{align}\label{Tm}
	T_{\lambda, n}^N (u, a)(g)=\begin{cases}
	  tr_{V_\lambda}\left(R_{01} (u-a_1)\dots  R_{0n}(u-a_n) (\pi_\lambda(g)\otimes Id^{\otimes n})\right),& \text{if} \quad l\le N,\\
	0,&\text{if} \quad l> N,
	\end{cases}
	\end{align}
 where the  trace is  taken over the  $0$-th component $V_\lambda$.  We think of these $T$-operators as functions from $GL_N(\bC)$  to the space of $\text{Hom}((\bC^N)^{\otimes n} )[1/u]$. We will omit $n, g, a$  and $N$ labels  and  write $T_\lambda(u)$  or $T^N_\lambda(u)$  when no confusion arrises.

 The  R-matirx  (\ref{Rm}) is an image of  the universal $R$-matrix of  the Yangian $Y(\fgl_N(\bC))$, and the  quantum Yang-Baxter equation on the universal  $R$-matrix  implies that    $T$-operators  (\ref{Tm}) commute for a fixed twist  matrix $g$,  a fixed collection  of parameters $a=(a_i)_{i=1,2,\dots}$ and  for all $u$
 and $\lambda$.
 
\begin{remark}\label{rem1}
 When  $n=0$, the element $T_\lambda(u)=tr\,{\pi_\lambda (g)}$  coincides with  the value of the character of  $\pi_\lambda$ on $g$ given by  Schur polynomial  $s_\lambda(g_1,\dots, g_N)$  of the eigenvalues of the matrix $g$. Also    $T_\lambda(u) \to  s_\lambda(g_1,\dots, g_N)\,Id^{\otimes n}$ when $u\to\infty$.
\end{remark}

\subsection {CBR determinant} Introduce the notation 
$ h_k(u)= T_{(k)}(u)$ for $k=1,2,\dots$.
It is  also    convenient   to set  $h_k(u)=0$ for $k=-1, -2,\dots$, and $h_0(u)=1$. 
 Then the  the following remarkable relation  between $T$-matrices was observed in  \cite{Baz},  \cite{Cher} (see also \cite{KNS}).
 \begin{theorem} [Cherednik-Bazhanov-Reshetikhin (CBR) determinant]
	\begin{align}\label{Tdet1}
	T_{\lambda}(u)=\det_{i,j=1,\dots,l} h_{\lambda_i-i+j}(u-j+1).
	\end{align}
\end{theorem}
This  formula is an  analogue of Jacobi~--~Trudi identity for Schur symmetric functions. 
It not only implies several  important  properties of  transfer matrices $T_\lambda(u)$,
but  leads to  the action of Clifford algebra of fermions on the linear space spanned by sequences of  $T$-matrices. 
For the later goal  we  follow the approach  developed in \cite{JR1}, \cite{JR2}.

\subsection{Newton's identity}
For any  $u$ set  $e_k(u)=T_{(1^k)}(u)$ for $k=1,2\dots$, and  $e_k(u)=0$ for $k=-1,-2$, as well as $e_0(u)=1$. 
Then from (\ref{Tdet1}) for $k=1,2,\dots$, 
\begin{align} \label{edet1}
e_k(u)= T_{(1)^k} (u)=\det_{i,j=1,\dots, k}[h_{1-i+j} (u-j+1)].
\end{align}

\begin{proposition}
a) The following analogue  of  Newton's formula relation  holds:
\begin{align}\label{Newt}
\sum_{p=-\infty} ^{+\infty}(-1)^{a-p} h_{b+p}(u-p) e_{-p-a} (u-p-1)=\delta_{a,b} \quad \text{for any  $a,b\in \bZ$.}
\end{align}
 b)  Let $\lambda ^\prime=(\lambda^\prime_1,\dots,\lambda^\prime_k)$  be the conjugate partition of the partition $\lambda$.
Then the dual (equivalent)  form of  (\ref{Tdet1}) holds:
\begin{align}\label{Tdet2}
T_\lambda(u)= \det_{i,j=1,\dots,k}[ e_{\lambda^\prime _i-i+j}(u+j-1)].
\end{align}
\end{proposition}
\begin{proof}
 Relation (\ref{Newt}) follows from the  recursive  expansion of the  determinant (\ref{edet1})  for  $e_{b-a} (u+b-1)$ by the  first  row.
 Then (\ref{Newt}) implies  that   the   upper  triangular infinite  matrices
 \begin{align*}
 \mathcal{H} = (\,h_{p-b} (u-p)\,)_{b,p\,\in \bZ},\quad 
 \mathcal{E} =(\, (-1)^{a-p} e_{a-p} (u-p-1)\,)_{p,a\,\in \bZ}
 \end{align*}
 satisfy   the identity $\mathcal{H}\mathcal{E}=Id$, and  
   (\ref{Tdet2})  follows from the  standard argument on  the  relations on  minors of matrices that are inverse of each other 
 (see e.g. Lemma A.42 of \cite{Fult}).
\end{proof}

\begin{remark}
For some applications it  is convenient to write  Newton's  formula  in  the form of an equation on generating functions:
$$
YH(u|t) YE(u-1|t)=1,
$$
where
$$
YH(u|t)= \sum_{s=0}^{\infty} (te^{-\partial_u})^s h_s(u),
\quad 
YE(u|t)= \sum_{s=0}^{\infty} (-1)^{s} e_s(u)(te^{-\partial_u})^s
$$
are  generating functions for  {\it operators} acting on the space of $\text{Hom}\,(\bC^N)^{\otimes n}$-valued  functions in variable $u$, and $e^{\partial_u} (f(u))= f(u+1)$ is the   a shift  operator of the variable $u$.  
Indeed,
\begin{align*}
YH(u|t) YE(u-1|t)=
\sum_{s,p\, =0}^{\infty} h_s(u-s)(-1)^p  e_p(u-s-1)(te^{-\partial_u})^{s+p}
\\
=\sum_m\left( \sum_{s=0}^{\infty} h_s(u-s)(-1)^{m-s}  e_{m-s}(u-s-1) \right) (te^{-\partial_u})^{m}=1.
\end{align*}

\end{remark}
\begin{remark}
From definition (\ref{Tm}),  it follows that infinite  matrix  $\mathcal{E}$ and infinite formal  series $YE(u|t)$ contain only finite number of non-zero terms  $e_k(u)= e^N_{k, n}(u,a ) (g)$. 
\end{remark}
\begin{remark}
In    \cite{MTV1} common eigenvalues and common  eigenvectors of higher transfer matrices  of an  XXX-type model are studied through a certain  generating function of transfer matrices with coefficients in  the Yangian  $Y(\fgl_N(\bC))$  (see formula (4.16) in \cite{MTV1}), and   earlier in \cite{Tal},  the higher 
 transfer matrices of the Gaudin model were constructed  explicitly as a limit of  that generating function.  Up to a  multiplication by a rational function of $u$, the generating function $YE(u|t)$  is the image of  the  generating function  considered in  \cite{MTV1}, \cite{Tal}  under the evaluation representation  $\bC_N(a_1)\otimes\dots\otimes \bC_N (a_n)$ of  $Y(\fgl_N(\bC))$.
\end{remark}

\subsection{ Transfer matrices labeled by integer vectors}\label{sec_int}
Formula (\ref{Tdet1}) allows us  to extend the notion of $T_\alpha(u)$  for any  integer vector $\alpha=(\alpha_1,\dots, \alpha_l)$, $\alpha_i\in \bZ$,
by literally  setting 
\begin{align}\label{hhdet}
T_{\alpha}(u)=\det [ h_{\alpha_i-i+j}(u-j+1)].
\end{align}
Note  that 
\begin{align}\label{transposition}
 T_{(\dots, \alpha_i, \alpha_{i+1},\dots)}(u)= - T_{(\dots, \,\alpha_{i+1} -1\,, \,\alpha_{i}+1\, ,\dots)}(u), 
\end{align}
 and  that $T_\alpha(u)=0$ whenever  $\alpha_i -i= \alpha_j-j$  for some $i, j$.

Similarly, we can use formula (\ref{Tdet2}) to extend the definition of transfer  matrices to integer vectors.
Set
\begin{align}\label{eedet1}
 T_{\alpha^\prime}(u)=\det_{i,j=1,\dots,k}[ e_{\alpha_i-i+j}(u+j-1)].
 \end{align}
 While we do not define   here  the conjugation on  arbitrary integer vector,  the notation $T_{\alpha^\prime}(u)$ for the expression (\ref{eedet1}) is justified by the following lemma. 
\begin{lemma}
Let $\alpha=(\alpha_1,\dots, \alpha_l)$ be an integer vector, let $\rho=\rho_l= (0,1,\dots,l-1)$. Let $T_\alpha(u)$ and $ T_{\alpha^\prime} (u)$  be defined by  (\ref{hhdet}), (\ref{eedet1}). Then
\begin{align}
 T_\alpha(u)=
 \begin{cases}
 (-1)^{\sigma} T_\lambda(u),&\quad \text{if} \quad \alpha-\rho=\sigma(\lambda-\rho)\,\text{for some  partition}\, \lambda \,\text{and}\, \sigma\in S_l, 
 \\
 0,&\text{otherwise.}
  \end{cases}
  \end{align}
 \begin{align}
  T_{\alpha^\prime}(u)=
 \begin{cases}
 (-1)^{\sigma} T_{\lambda^\prime}(u),&\quad \text{if} \quad \alpha-\rho=\sigma(\lambda-\rho)\,\text{for some  partition}\, \lambda \,\text{and}\, \sigma\in S_l, 
 \\
 0,&\text{otherwise.}
 \end{cases}
\end{align}
\end{lemma}

\begin{figure}
\includegraphics[width=50mm]{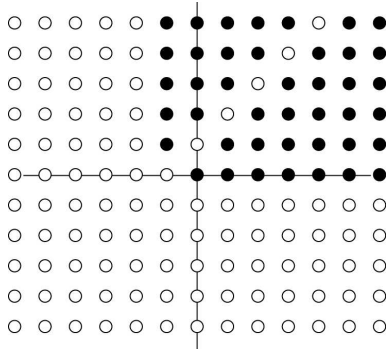}\\
\caption{Non-zero values of  $T_{(\alpha_1,\alpha_2)}(u)$ for $N\ge 2$.}
\label{fig:1}       
\end{figure}
Figure \ref{fig:1} illustrates the distribution of non-zero values of $T$-operators for integer vectors $(\alpha_1, \alpha_2)$ for $N\ge 2$.
Black  points   represent integer vectors $(\alpha_1,\alpha_2)$  with non-vanishing   transfer matrices $T_{(\alpha_1,\alpha_2)}(u)$.

\subsection{Generating  functions of $T$-operators} 
Set
 \begin{align*}
 H^N(x| u)= \sum_{k=0}^{\infty} h^N_k(u) x^k, \quad 
  E^N(x| u)= \sum_{k=0}^{\infty}(-1)^k e^N_k(u) x^k.
 \end{align*}
 \begin{proposition}
For any  integer  vector $\alpha$, the  transfer matrix $T^N_\alpha(u)$  is the coefficient of $x_1^{\alpha_1}\dots x_l^{\alpha_l-l+1}$ in
\begin{align}\label{H11}
H^N(x_1,\dots, x_l|u)&=\det \left[{x_i^{-j+1}}H^N(x_i|u-j+1)\right]\\
&=\left.\left(\prod_{1\le i<j\le l}\left( \frac{ e^{-\partial_{u_j} }  }  {x_j}-\frac{ e^{-\partial_{u_i}}} {x_i}\right)\prod_{i=1}^{l}H^N(x_i|u_i)\right)\right\vert_{u_1=u_2=\dots= u_l=u}.\notag
\end{align}
Similarly,  $ (-1)^{A}T^N_{\alpha^\prime}(u)$  with $A=\sum(\alpha_i-i+1)$ is the coefficient of $x_1^{\alpha_1}\dots x_l^{\alpha_l-l+1}$ in
\begin{align}\label {E11}
E^N(x_1,\dots, x_l|u)&=\det \left[{(-x_i)^{-j+1}}E^N(x_i|u+j-1)\right]\\
&=\left.\left(\prod_{1\le i<j\le l}\left( \frac{ e^{\partial_{u_i} }  }  {x_i}-\frac{ e^{\partial_{u_j}}} {x_j}\right)\prod_{i=1}^{l}E^N(x_i|u_i) \right)\right\vert_{u_1=u_2=\dots= u_l}.
\notag
\end{align}

 \end{proposition}

\begin{proof}
The first statement of (\ref{H11}) follows from the  expansion  of the determinant as a sum over permutations:
\begin{align*}
&\sum_{\alpha\in \bZ^{l}}T^N_\alpha(u) x_1^{\alpha_1}\dots x_l^{\alpha_l-l+1}=
\sum_{\alpha\in \bZ^{l}}\det [h^N_{\alpha_i-i+j}(u-j+1)] x_1^{\alpha_1}\dots x_l^{\alpha_l-l+1}\\
&= \sum_{\alpha\in \bZ^{l}}\sum_{\sigma\in S_l} (-1)^{\sigma}h^N_{\alpha_1-1+\sigma(1)}(u-\sigma(1)+1)\dots h^N_{\alpha_l-l+\sigma(l)}(u-\sigma(l)+1) x_1^{\alpha_1}\dots x_l^{\alpha_l-l+1}\\
&=\sum_{\sigma\in S_l}\sum_{(a_1,\dots, a_l)\in \bZ^{l}} (-1)^{\sigma}h^N_{a_1}(u-\sigma(1)+1)\dots h^N_{a_l}(u-\sigma(l)+1) x_1^{a_1-\sigma(1)+1}\dots x_l^{a_l-\sigma(l)+1}\\
&=\sum_{\sigma\in S_l} (-1)^{\sigma}H^N(x_1|u-\sigma(1)+1)\dots H^N(x_l|u-\sigma(l)+1) x_1^{-\sigma(1)+1}\dots x_l^{-\sigma(l)+1}\\
&=\det \left[{x_i^{-j+1}}H^N(x_i|u-j+1)\right].
\end{align*}
For the second part of (\ref{H11}), which can be  viewed as a  generalization of (\ref{fsym}),
consider a set of independent  variables  $(u_1,\dots, u_l)$. Define
$$H^N(x_1,\dots, x_l|u_1,\dots,u_l)=\det \left[{x_i^{-j+1}}H^N(x_i|u_i-j+1)\right].$$ 

Then
\begin{align*}
H^N(x_1,\dots, x_l|u_1,\dots  u_l)
=\sum_{\sigma\in S_l}(-1)^{\sigma}\left({x_1}{ e^{\partial_{u_1}}}\right) ^{1-\sigma(1)}H^N(x_1|u_1) \dots\left( {x_l}e^{\partial_{u_l}}\right) ^{1-\sigma(l)}H^N(x_1|u_1) 
\\
=\det[  (x_i e^{\partial_{u_i}}) ^{1-j}] \prod_{i=1}^{l}H^N(x_i|u_i)
=\prod_{1\le i<j\le l}\left( \frac{ e^{-\partial_{u_j} }  }  {x_j}-\frac{ e^{-\partial_{u_i}}} {x_i}\right)\prod_{i=1}^{l}H^N(x_i|u_i),
\end{align*}
and the second  statement follows. 
The  proof of (\ref{E11}) is exactly the same. 
\end{proof}

\section{Clifford algebra action on the space of  sequences of transfer matrices } 
\subsection{Clifford algebra action on symmetric functions}
Our next goal is to  define the action of Clifford algebra on a space generated by transfer matrices, which will allow us to establish  connections to the classical  boson-fermion correspondence and $\tau$-function formalism. 
Clifford algebra  $Cl$ is generated by elements $\{\psi^{\pm}_k\}_{k\in \bZ}$  with the relations
\begin{align}
\psi^+_k\psi^-_l+\psi^-_l\psi^+_k=\delta_{k,l},\quad
\psi^\pm_k\psi^\pm_l+\psi^\pm_l\psi^\pm_k=0. \label{pmrel}
\end{align}
This infinite   Clifford algebra   originates from  the vector space  $W=(\oplus_i\bC\psi^+_i)\bigoplus (\oplus_i\bC\psi^-_i)$  with the bilinear form  that is defined by  $(\psi^+_i,\psi^-_j)=\delta_{ij}$, and the rest of the values being zeros.

The crucial component  of the classical  boson-fermion correspondence   refers to the action  of $Cl$ on the  space of symmetric functions, which we reproduce here in purely combinatorial terms.
Let  $z$  be a formal  variable and let   $\B^{(m)}$ be the  linear  span of  the set  of  basis  vectors $\{ z^ms_\lambda\}$, where  $s_\lambda$ are  Schur symmetric functions  taken over all  partitions  $\lambda$. Note that each $\B^{(m)}$ is just  a copy of the ring of symmetric functions. Set  
$$
\B= \oplus_{m\in\bZ} \B^{(m)}.
$$
The action of  operators $\psi^\pm_{k}$ on   $\{z^m s_\lambda\}$  is given by the following rules:
\begin{align}
\psi^+_{k} (z^ms_\lambda) &= z^{m+1}s_{(k-m-1,\lambda)} ,\label{p10}\\ 
\psi^-_k (z^ms_\lambda)&=z^{m-1}\sum_{t=1}^{\infty}(-1)^{t+1}\delta_{k-m-1,\,\lambda_t-t}\,  s_{(\lambda_1+1,\dots, \lambda_{t-1}+1, \lambda_{t+1},\lambda_{t+2}\dots) },  \label{p100} 
\end{align}
where again  for any integer vector $\alpha$ of $l$ parts, one sets  $s_\alpha= (-1)^{\sigma}s_\lambda$, if $\alpha-\rho= \sigma (\lambda-\rho)$ for some partition $\lambda$  and some permutation  $\sigma$, and $s_\alpha=0$ otherwise. 
Note that   only one term survives in the sum (\ref{p100}).  One of the easiest ways to check that  operators $\psi^\pm_{k}$ satisfy the relations of Clifford algebra is through  identification of basis elements $z^m s_\lambda$ with  semi-infinite  monomials of  fermionic Fock space 
$$
z^m s_\lambda \sim v_{m+\lambda_1}\wedge v_{m-1+\lambda_2}\wedge  v_{m-2+\lambda_3}\wedge\dots.
$$
Here $\{v_i\}_{i\in \bZ}$ is a linear basis of a vector space $V=\oplus_{i\in \bZ}\bC v_i$.  
Operators   $\psi^\pm_{k}$  are creation and annihilation operators on the linear span of such monomials (see e.g. \cite{D1}, \cite{K} or \cite{Bom}  for more details).

Formulae  (\ref{Tdet1}), (\ref{Tdet2}) suggest that we can define Clifford algebra action  on the space of  transfer matrices as well. This construction is described below, and first, we would like to explain the necessity of  some adjustments. Recall that  definition (\ref{Tm})  of transfer matrix $T_\lambda(u)=T_{\lambda, n}^N (u, a)(g)$ depends on the size $N$  of an invertible  matrix $g$. For consistency   we are forced to set $T_{\lambda,n}^N(u,a)(g)=0$   whenever the number $l$ of parts of $\lambda$
 is greater than $N$. In that sense, transfer matrices 
$T_{\lambda,n}^N(u, a) (g)$  create  natural analogue of Schur symmetric polynomials rather  than  Schur symmetric functions, which is not  yet enough to construct a $Cl$-module of the type $\B$.  Hence,  to get  the analogue of Schur symmetric functions, we  have to consider the sequences of  the transfer matrices 
$\{T_{\lambda, n}^N (u, a)(g)\}_{N=1}^{\infty}$. Since we deal only with  linear vector space structure, the questions of stability  of these  sequences, important for the ring structure of symmetric functions, seem to be of  no  concern for us here. 
\subsection{Clifford algebra action on the space of sequences of transfer matrices.}
Let  $n\in \bN$  and  $a=(a_1,\dots, a_n)$ be fixed parameters. As above,   $h_k(u)=h_{k,n}^N(u,a)$  is  a function from $GL_N(\bC)$ to $\text{Hom}( \bC^N)^{\otimes n}\left[1/u\right]$, defined by  
\begin{align}
h_{k,n}^N(u,a) (g)=tr_{Sym^k\bC^N} (R_{01} (u-a_1)\dots R_{0n} (u-a_n) Sym^k(g)\otimes Id^n),
\end{align} 
with $R$-matrix  (\ref{Rm}) for $V_\lambda= Sym^k\bC^N$.
Recall that for  any partition  $\lambda= (\lambda_1\ge\lambda_2\ge\dots\ge\lambda_l>0)$,
\begin{align*}
	T^N_{\lambda,n}(u,a) (g)=\begin{cases}
	\det_{i,j=1,\dots,l} [h^N_{\lambda_i-i+j,n}(u-j+1,a)(g)] , & \text{if} \quad l\le N,\\
	0,&\text{if} \quad l> N
	\end{cases}
	\end{align*}
	defines   a function $T^N_{\lambda,n}(u,a)$ from $GL_N(\bC)$ to $\text{Hom}( \bC^N)^{\otimes n}\left[1/u\right]$. 
	
Set ${\bf t}_\lambda (u)$ to be the sequence of  functions  $T^N_{\lambda,n}(u,a)$ 	
\begin{align*}
{\bf t}_\lambda(u)=\,(T^N_{\lambda,n}(u,a))_{N=1}^{\infty}\,= \,(0, \dots, 0,\, T^{l}_{\lambda,n}(u,a),\, T^{l+1}_{\lambda,n}(u,a),\,...).
\end{align*}

In the view of the  Remark \ref{rem1},  ${\bf t}_\lambda(u)$ form a set  of linearly independent elements. 
We set $\tilde \B^{(0)}$
to be the linear span of all  ${\bf t}_\lambda(u)$ over all partitions $\lambda$. Set
$
\tilde \B^{(m)} = z^m\tilde \B^{(0)} 
$
for $m\in\bZ$
and $\tilde \B=\bigoplus_{m\in\bZ}\tilde \B^{(m)} $
Along the lines of Section \ref{sec_int}, we extend the definition of sequences of transfer matrices to the ones labeled by integer vectors of $l$ parts  by setting
\begin{align*}
{\bf t}_\alpha(u)=
 \begin{cases}
 (-1)^{\sigma} {\bf t}_\lambda(u),&\quad \text{if} \quad \alpha-\rho=\sigma(\lambda-\rho)\,\text{for a partition}\, \lambda \,\text{and }\, \sigma\in S_l, 
 \\
 0,&\text{otherwise.}
  \end{cases}
  \end{align*} 
  We have: 
\begin{align}\label{tr}
{\bf t}_{(\dots, \alpha_i, \alpha_{i+1},\dots)}(u)= - {\bf t}_{(\dots, \,\alpha_{i+1} -1\,, \,\alpha_{i}+1\, ,\dots)}(u).
\end{align}
We also extend the conjugation  operation on the sequences  labeled by  integer vectors:
\begin{align}
 {\bf t}_{\alpha^\prime}(u)=
 \begin{cases}
 (-1)^{\sigma} {\bf t}_{\lambda^\prime}(u),&\quad \text{if} \quad \alpha-\rho=\sigma(\lambda-\rho)\,\text{for a partition}\, \lambda \,\text{and }\, \sigma\in S_l, 
 \\
 0,&\text{otherwise.}
  \end{cases}
  \end{align}

Then the action of $Cl$  on   $\tilde \B$ can be defined exactly by the same formulae as  the action of this algebra on the boson space of symmetric functions:
\begin{align}
\psi^+_{k} (z^m{\bf t}_\lambda (u)) &= z^{m+1}{\bf t}_{(k-m-1,\lambda)}(u) ,\label{Tcl1}\\ 
\psi^-_k (z^m{\bf t}_\lambda(u) )&=z^{m-1}\sum_{s=1}^{\infty}(-1)^{s+1}\delta_{k-m-1,\,\lambda_s-s}\, {\bf t}_{(\lambda_1+1,\dots, \lambda_{s-1}+1, \lambda_{s+1},\lambda_{s+2}\dots) } (u).  \label{Tcl2} 
\end{align}
The property (\ref{tr}) implies that this is a well-defined action of  $Cl$ on $\tilde\B$, and we have an isomorphism  of $Cl$-modules.
\begin{proposition}
The map from $\varphi: \B \to \tilde \B$, defined on the basis vectors  $\varphi(s_\lambda z^m)= t_\lambda(u)z^m$ defines the $Cl$-module isomorphism.
\end{proposition}

These statements follow immediately.
\begin{proposition} 
  (a) Let $ 1=(1, 1,1,\dots)\in \tilde \B ^{(0)}$ be the vacuum vector of the graded component $\tilde \B ^{(0)}$ of the bosonic space of sequences of transfer matrices, and let $\lambda=(\lambda_1\ge\lambda_2\ge\dots \ge \lambda_l\ge0)$ be a partition.  
Then 
\begin{align*}
\psi^+_{\lambda_1+l}\dots \psi^+_{\lambda_l+1} (1)&= z^{l} {\bf t}_\lambda,\quad \\
\psi^-_{-\lambda_l-l+1}\dots \psi^-_{-\lambda_1} (1)&= z^{-l} (-1)^{\sum\lambda_i} {\bf t}_{(1^{(\lambda_1-\lambda_2)},2^{(\lambda_2-\lambda_3)},\,\dots\,, l^{\lambda_l})},
\end{align*}
where we use $(1^{(a_1)},2^{(a_2)}\dots)$ notation for a partition that has $a_1$ parts of length $1$, $a_2$ parts of length $2$, etc. 

 (b)  Let $\tau=\sum_\lambda c_\lambda s_\lambda $   be a solution of 
\begin{align}\label {H1}
\sum_{k\in \bZ}\psi^+_k( \tau)\otimes \psi^{-}_{k} (\tau)=0
\end{align}
for the  Clifford algebra action on the boson space $\B$ of symmetric functions. 
Then $ {\tilde \tau}  =\sum_\lambda c_\lambda {\bf t}_\lambda (u)= \left(\sum_\lambda c_\lambda T^N_{\lambda, n} (u,a)\right)_{N=1}^{\infty}$  is a solution of the  equation 
$$
\sum_{k\in \bZ}\psi^+_k( \tilde \tau (u))\otimes \psi^-_{k} (\tilde\tau(v))=0
$$
for  the  Clifford algebra action on the  boson space $\tilde \B$  of sequences of transfer matrices. (cf.   \cite{Z1}, \cite{Z3}, \cite{Z2},  \cite{Z4}).
\end{proposition}

 \section{Notes on bosonisation}

Combine operators  $\psi^{\pm}_k$  into generating  functions
 \begin{align} \label{pp}
\Psi^+(x,m)=\sum_{k\in \bZ}^\infty \psi^+_k|_{ \B^{(m)}}x^{k},
\quad 
\Psi^-(x,m)=\sum_{k\in \bZ}^\infty \psi^-_k|_{ \B^{(m)}}x^{-k}.
\end{align}
Then the relations of Clifford algebra are 
\begin{align*} 
&\Psi^\pm(x, m\pm1)\Psi^\pm(y, m)+\Psi^\pm(y, m\pm1)\Psi^\pm(x, m)=0,\\
&\Psi^-(x, m+1)\Psi^+(y, m)+\Psi^+(y, m-1)\Psi^-(x, m)= \delta( x^{-1}, y^{-1}),
\end{align*}
with formal  distribution 
$
\delta(x^{-1},y^{-1})=\sum_{k\in \bZ}\frac{y^{k+1}}{x^k}.
$
 It is known (see e.g. proof in \cite{JR2}) that  the action of $\Psi^\pm(x,m)$ on vacuum vector $1$, the constant polynomial  in the boson space $\B^{(0)}$ of Schur functions, produces the generating functions of symmetric functions:
 \begin{align}
 \Psi^+(x_1, l-1)\dots \Psi^+(x_l, 0) (1)&= z^{l} x_1^l\dots x_l^l \sum s_\lambda x_1^{\lambda_1}\dots x_l^{\lambda_l+1-l},\\ 
  \Psi^-(x_1,- l+1)\dots \Psi^-(x_l, 0) (1)&= z^{-l} x_1^l\dots x_l^l \sum (-1)^A s_\lambda x_1^{\lambda_1}\dots x_l^{\lambda_l+1-l},
 \end{align}
where $A=\lambda_1+\dots +\lambda_l$.  Clifford algebra modules isomorphism immediately implies analogous statement for the action of these operators on $\tilde\B$.
\begin{proposition}
Let ${\bf 1}= (1)_{N=1}^{\infty}$ be the vacuum vector in  in the boson space $\tilde \B^{(0)}$ of sequences of transfer matrices. Then 
 \begin{align}
 \Psi^+(x_1, l-1)\dots \Psi^+(x_l, 0) ({\bf 1})&= \left(z^{l}\, x_1^l\dots x_l^l \,H^N(x_1,\dots,x_l|u) \right)_{N=1}^{\infty}\\ 
  \Psi^-(x_1,- l+1)\dots \Psi^-(x_l, 0) ({\bf1})&=\left( z^{-l} \, (-x_1)^l\dots  (-x_l)^l\, E^N(x_1,\dots, x_l|u)\right)_{N=1}^{\infty}, 
 \end{align}
 where $\Psi^\pm(x, m)$ are generating functions of the action of operators $\psi^{\pm}$ on the graded component $\tilde \B^{(m)}$.
\end{proposition}

 Recall that the transition between (\ref{H1}) and  Hirota bilinear  equations that  produce KP hierarchy  is based on the bosonization of the operators  $\Psi^{\pm} (x,0)$.  
Namely, consider the action of $\Psi^{\pm} (x,0)$ on the boson space of symmetric functions.  Let   $h_r= s_{(r)}$  be complete symmetric functions, let  $e_r= s_{(1^r)}$ be elementary symmetric functions, 
and let $p_k$  be (normalized) classical power sums.
One can  write generating  functions  for  complete and elementary symmetric functions in the form
$H(x)=\sum_{k\ge 0} h_k x^k,$
$E(x)=\sum_{k\ge 0} (-1)^ke_k x^k$.
Then the  families $\{e_r\}, \{h_r\}, \{p_r\}$  are related by the identities  
 \begin{align*}
 H(x) E(x)=1,\quad
 H(x)=\exp\left(\sum_{k\ge 1}p_kx^k\right),\quad 
E(x)=\exp\left(-\sum_{k\ge 1} p_k x^k\right).
\end{align*}
The
ring of symmetric functions  possesses a natural  scalar product, where the classical Schur functions $s_\lambda$ constitute
  an orthonormal basis $\<s_\lambda,s_\mu\>=\delta_{\lambda, \mu}$. For any symmetric function $f$ one can
define the adjoint operator $D_f$ acting on the ring of symmetric functions by the standard rule:
$ \<D_fg, w\>=\<g,fw\>$, where $g,w$ are any symmetric functions.  The properties of  adjoint operators are described  e.g. in  \cite{Md},  Section I.5. Set
  \begin{align*}
  DE(x)= \sum_{k\ge 0}(-1)^k D_{e_k} x^{-k},\quad DH(x)= \sum_{k\ge 0} D_{h_k} x^{-k}.
  \end{align*}
     Then   the operators $ \Psi^{\pm}(x,m)$ acting on the component  $\B^{(m)}$  can be written in the vertex operator form 
\begin{align}
    \Psi^+(x,m)&= x^{m+1}z \,H(x) DE\left(x\right),\quad 
       \Psi^- (x,m)= x^{{-m}}z^{-1} \, E(x) DH\left(x\right).\label{psi11}
\end{align}
Often these formulae are  written through the power sums  (see e.g. \cite{Bom} Lecture 5, or  \cite{Md} Section I.5 Example 29). Any  symmetric function  $f$ can  be expressed as a polynomial   $ f=\varphi(p_1,2p_2, 3p_3,\dots)$ in (normalized) power sums. Then one has
$D_f=\varphi ({\partial_{ p_1}},  {\partial_ {p_2}}, {\partial_{ p_3}},\dots)$.
 (See e.g. \cite {Md}, I.5, Example 3).
Hence
$$
DH(x)=\exp\left(\sum_k {\frac {\partial_{ p_k}}{k}}x^{-k}\right), \quad  DE(x)= \exp\left(-\sum_k \frac{\partial_{ p_k}}{k}x^{-k}\right),
$$
and we can write
\begin{align}
     \Psi^+(x,m)= x^{m+1}z\exp \left(\sum_{j\ge 1}p_j x^j  \right) \exp \left(-\sum_{j\ge 1}\frac{ \partial_{p_j}}{j} {x^{-j}}\right),
\label{psi2}\\
     \Psi^-(x,m)= x^{-m}z^{-1} \exp \left(-\sum_{j\ge 1} {p_j} x^j  \right) \exp \left(\sum_{j\ge 1}\frac{\partial_{p_j}}{j}  {x^{-j}}\right).
\label{psi21}
\end{align}
The formulae  (\ref{psi2}), (\ref{psi21}) describe the operators  in terms of  generators  $\{p_i, \partial_{p_j}\}_{j=1,2,\dots}$ of the Heisenberg algebra acting on the boson space. The substitution of variables and the Taylor expansion of these formulas produce the Hirota bilinear  equations  of KP hierarchy \cite{D1}, \cite{Hirota1},\cite{Miwa1}.

The constructed in Section 3   isomorphism  $\B\simeq \tilde \B$ of Clifford algebra modules does not transport the ring structures of these spaces, hence  decompositions  (\ref{psi11}),
(\ref{psi2}), (\ref{psi21}) are not  applicable to $ \Psi^\pm(u,m)$-action on $\tilde \B$.
The partial analogue  of  (\ref{psi11}) for  the components of the action of  $Cl$ on the sequences of transfer matrices is stated below, but  at the  moment we do not know  further  interpretation of vertex operators   through Heisenberg-type generators  in the spirit of  (\ref{psi2}), (\ref{psi21}).

Recall determinant formulae (\ref{H11}),  (\ref{E11})  for  generating functions  $H^N(x_1, x_2\dots, x_l|u)$ and  $E^N(x_1, x_2\dots, x_l|u)$.
Since operators  $x_i e^{-\partial_{u_i}}$  commute with each other, one can apply  the expansion
\begin{align*}
\prod _{i=1}^{m}(a_i-y)=\sum_{k=0}^{m} (-y)^{m-k}  e_k(a_1,\dots, a_m)
\end{align*}
in terms of elementary symmetric functions $e_k(a_1,\dots, a_m)$ 
to the product:
\begin{align*}
\prod _{i=1}^{l}\left(\frac{e^{-\partial_{u_1}}}{x_1}-\frac{e^{-\partial_{u}}}{x}\right)=
\sum_{k=0}^{l} \left(-\frac{e^{-\partial_{u}}}{x}\right)^{l-k}  e_k\left(\frac{e^{-\partial_{u_1}}}{x_1},\dots, \frac{e^{-\partial_{u_l}}}{x_l}\right).
\end{align*}
Hence, from (\ref{H11}),
\begin{align}
H^N(x, x_1,\dots, x_l|u)=\sum_{k=0} ^{l} 
(- x)^{-l+k} H^N(x|u-l+k) DE^N_{k} \, H(x_1,\dots, x_l|u),\\
E^N(x, x_1,\dots, x_l|u)=\sum_{k=0} ^{l} 
(- x)^{-l+k} H^N(x|u+k-l) DH^N_{k} \, E(x_1,\dots, x_l|u),
\end{align}
where $ DE^N_{k}$,  $ DH^N_{k}$  are  analogues of  operators $D_{e_k}$,  $D_{h_k}$. They act on the coefficients of the  generating series of transfer matrices  by  formulae
\begin{align*}
 DE^N_{k} \, H^N(x_1,\dots, x_l|u)=
e_k\left(\frac{e^{-\partial_{u_1}}}{x_1},\dots, \frac{e^{-\partial_{u_l}}}{x_l}\right)H( x_1,\dots, x_l|u_2,\dots  u_l)\,
 |_{u_1=\dots=u_l=u},\\
  DH^N_{k} \, E^N(x_1,\dots, x_l|u)=
e_k\left(\frac{e^{\partial_{u_1}}}{x_1},\dots, \frac{e^{\partial_{u_l}}}{x_l}\right)E( x_1,\dots, x_l|u_2,\dots  u_l)\,
 |_{u_1=\dots=u_l=u}.
\end{align*}
This sums up to the following statement of the  decomposition of coordinates of the action of  generating functions $\Psi^\pm(x,m)$  on the sequences of transfer matrices.
\begin{proposition}
Let $\Psi^\pm(x,m)$ be generating functions of the restrictions  of action  of $\psi^{\pm}_k$ on the graded component $\tilde \B^{(m)}$ (cf. (\ref{pp})).
Let $\lambda=(\lambda_1,\dots, \lambda_l)$  be a partition of at most $l$  parts. Then the coordinates of the action of $\Psi^\pm(x,m)$ 
on $ z^m{\bf t_\lambda}=(z^m T^N_{\lambda, n}(u,a))_{N=0}^{\infty}$ can be decomposed as
\begin{align*}
\Psi^+(x,m)(z^m{\bf t_\lambda})= \left(z^{m+1}\sum_{k=0} ^{l} 
(- x)^{-l+k} H^N(x|u-l+k) DE^N_{k}\,  T^N_{\lambda, n}(u,a)\right)_{N=1,...,\infty},\\
\Psi^-(x,m)(z^m{\bf t_\lambda})= \left(z^{m-1}\sum_{k=0} ^{l} 
(- x)^{-l+k} E^N(x|u+l-k) DH^N_{k}\,  T^N_{\lambda, n}(u,a)\right)_{N=1,...,\infty}.
\end{align*}
\end{proposition}
\section{Aknowledgements} The  author would like to thank E.~Mukhin for  valuable discussions and the reviewer for consideration of the text.

\end{document}